\documentclass[9pt,shortpaper,twoside,web]{ieeecolor}
\usepackage{generic}
\usepackage{textcomp}
\def\BibTeX{{\rm B\kern-.05em{\sc i\kern-.025em b}\kern-.08em
    T\kern-.1667em\lower.7ex\hbox{E}\kern-.125emX}}
\usepackage[noadjust]{cite}
\usepackage{amsmath,amssymb,amsfonts}
\usepackage{bigints}
\usepackage{abraces}
\usepackage{mathrsfs}
\usepackage{color}
\usepackage{graphicx}
\usepackage{epsfig} % for postscript graphics files
\usepackage{hyperref}
\usepackage{subcaption}
\usepackage{caption}

\usepackage{array,booktabs} % For nice tables
\usepackage{amsmath,amssymb} % For nice maths
\usepackage{stfloats} % For putting double column floats on bottom of page
\usepackage{psfrag} % For putting double column floats on bottom of page
\usepackage{color}

\definecolor{box_color}{rgb}{.8,.8,.8}
\usepackage[utf8]{inputenc}
\usepackage[T1]{fontenc}
\usepackage{epstopdf}

\newtheorem{proposition}{Proposition}
\newtheorem{corollary}{Corollary}
\newtheorem{fact}{Fact}
\newtheorem{definition}{Definition}
\newtheorem{remark}{Remark}
\newtheorem{example}{Example}

\usepackage{tikz,xcolor,hyperref}

\definecolor{lime}{HTML}{A6CE39}
\DeclareRobustCommand{\orcidicon}{
	\begin{tikzpicture}
	\draw[lime, fill=lime] (0,0) 
	circle [radius=0.16] 
	node[white] {{\fontfamily{qag}\selectfont \tiny ID}};
	\draw[white, fill=white] (-0.0625,0.095) 
	circle [radius=0.007];
	\end{tikzpicture}
	\hspace{-2mm}
}

\foreach \x in {A, ..., Z}{\expandafter\xdef\csname
orcid\x\endcsname{\noexpand\href{https://orcid.org/\csname
orcidauthor\x\endcsname}
			{\noexpand\orcidicon}}
}

\def\begquo{\begin{quote}}
\def\endquo{\end{quote}}
\def\begequarr{\begin{eqnarray}}
\def\endequarr{\end{eqnarray}}
\def\begequarrs{\begin{eqnarray*}}
\def\endequarrs{\end{eqnarray*}}
\def\begarr{\begin{array}}
\def\endarr{\end{array}}
\def\begequ{\begin{equation}}
\def\endequ{\end{equation}}
\def\lab{\label}
\def\begdes{\begin{description}}
\def\enddes{\end{description}}
\def\begenu{\begin{enumerate}}
\def\begite{\begin{itemize}}
\def\endite{\end{itemize}}
\def\endenu{\end{enumerate}}

\def\lef[{\left[\begin{array}}
\def\rig]{\end{array}\right]}
\def\begcen{\begin{center}}
\def\endcen{\end{center}}
\def\begfac{\begin{fact}}
\def\endfac{\end{fact}}
\def\begsubequ{\begin{subequations}}
\def\endsubequ{\end{subequations}}

\def\begmat#1{\begin{bmatrix}#1\end{bmatrix}}

\def\vex{\mbox{vex}}
\def\skew{\mbox{skew}}

\def\cale{{\cal E}}

\def\calj{{\cal J}}

\def\cale{{\cal E}}

\def\calj{{\cal J}}

\def\tr{\mbox{tr}}

\def\L2e{{\cal L}_{2e}}

\def\rea{\mathbb{R}}

\def\diag{\mbox{diag}}

\def\col{\mbox{col}}
\def\hal{{1 \over 2}}
\def\et{\varepsilon_t}
\def\diag{\mbox{diag}}

\usepackage{color}

\def\QED{\hfill $\square$}

\usepackage[prependcaption,colorinlistoftodos]{todonotes}

\input epsf

\def\BibTeX{{\rm B\kern-.05em{\sc i\kern-.025em b}\kern-.08em
    T\kern-.1667em\lower.7ex\hbox{E}\kern-.125emX}}
\begin{document}

\title{Attitude Estimation from Vector Measurements: Necessary and Sufficient
Conditions and Convergent Observer Design}
\author{Bowen Yi\orcidA{}, Lei Wang\orcidC, and Ian R. Manchester\orcidB
%\thanks{This paragraph of the first footnote will contain the date on  }
%
\thanks{This work was supported by Australian Research Council.}
\thanks{B. Yi and I. R. Manchester are with the Australian Centre for Field Robotics, and Sydney Institute for Robotics and Intelligent Systems, The University of
Sydney, Sydney, NSW 2006, Australia (email:
\texttt{bowen.yi,~ian.manchester@sydney.edu.au})
}
\thanks{L. Wang is with State Key Laboratory of Industrial Control Technology,
Institute of Cyber-Systems and Control, Zhejiang University, Hangzhou 310027,
China (email: \texttt{lei.wangzju@zju.edu.cn}, corresponding author)}
}

\maketitle
\thispagestyle{empty}

\begin{abstract}
The paper addresses the problem of attitude estimation for rigid bodies using (possibly time-varying) vector measurements, for which we provide a \emph{necessary and sufficient} condition of distinguishability. Such a condition is shown to be strictly weaker than those previously used for attitude observer design. Thereafter, we show that even for the single vector case the resulting condition is sufficient to design almost globally convergent attitude observers, and two explicit designs are obtained. To overcome the weak excitation issue, the first design employs to make full use of historical information, whereas the second scheme dynamically generates a virtual reference vector, which remains non-collinear to the given vector measurement. Simulation results illustrate the accurate estimation despite noisy measurements. 
\end{abstract}


%%%%%%%%%%%%%%%%%%%%%
\section{Introduction}
\label{sec1}
%%%%%%%%%%%%%%%%%%%%%

The attitude of a rigid body is its orientation with respect to an inertial reference frame. Attitude estimation is an essential element in a wide range
of robotics and aerospace applications, in particular control,
navigation, and localization tasks. Many common sensor types, e.g. magnetometers, accelerometers, and monocular cameras, provide body-fixed measurements of quantities with a known inertial value, e.g. the earth's magnetic field and gravitational force, or the bearing to certain known landmarks. These are known as \textit{complementary} measurements \cite{TRUetal}. In some less common scenarios a set of known vectors in the body-fixed frame are measured in the inertial frame, e.g. measurements from two GPS receivers attached  to the rigid body  with a  known base-line. These  are  known  as \textit{compatible} measurements \cite{TRUetal}.

Estimation of attitude from  multiple non-collinear vector measurements was formulated as a total least-squares problem over rotation matrices by Wahba \cite{WAH}. Several efficient algorithms exist for its solution, including singular value decomposition methods, TRIAD, and QUEST \cite{SHUOH}. 

However, when estimating a time-varying attitude it often is beneficial to fuse the vector measurements with information from gyroscopes using a dynamical model. The resulting dynamic estimator is commonly known as a filter or observer. These approaches can significantly reduce the impact of high-frequency measurement noise. Furthermore, in many applications there is only a single vector available for attitude estimation and in this case the attitude is not completely determined at a single moment. Applications for estimation from a single vector measurement include Sun sensors in eclipse periods \cite{NAMSAF}, improving reliability with redundant measurements and simplifying  designs \cite{REIetal}, as well visual-inertial navigation with only two feature
points visible in some periods.

Among  filtering approaches, extended Kalman filter is the most widely-applied for attitude estimation. However the domain of attraction is intrinsically local since the filter is based on first-order linearization; see \cite{CRAetal} for a recent review. Alternatively,  interest in nonlinear attitude observers was spurred by Salcudean's seminal work \cite{SAL}, and has achieved significant progress since then. There are many nonlinear attitude observers making direct use of vector measurements, e.g., with multiple measurements \cite{MAHetal08,TRUetal,ZLOFOR} or single vector measurements \cite{BATetal,GRIetalTAC,BAHNAM,KINWHI}. The latter works impose a persistently non-constant condition on the single reference vector, or similar conditions in which the \emph{uniformity} of excitation with respect to time plays an essential role to guarantee asymptotic convergence. In \cite{TRUetal}, the authors provide a comprehensive treatment of observability of a rigid-body attitude kinematic model with vectorial outputs. However, as illustrated in \cite[Remark 3.9]{TRUetal}, the condition is only sufficient but \emph{not} necessary for distinguishability, a specific type of observability for nonlinear dynamical systems \cite{BES,BER}. In this paper, we revisit the problem of observability analysis and propose two novel attitude observers. To be precise, the main contributions of the note are two-fold:

\begin{itemize}
\item[\bf C1] For the problem of attitude estimation from vector measurements, we provide the \emph{necessary and sufficient} condition to distinguishability of the associated dynamical model, which is known as the necessity to reconstruct attitude over time in any deterministic estimators;

\item[\bf C2] We show that the resulting distinguishability condition is also sufficient to design a continuous-time attitude observer. By focusing on single vector measurements, we provide two novel almost globally convergent attitude observers, which  require significantly weaker conditions than existing methods.
\end{itemize}

The constructive tool we adopt in observer design is the parameter
estimation-based observer (PEBO), which was recently proposed in Euclidean
space \cite{ORTetalscl,ORTetalAUT}, and extended to matrix Lie groups by
the authors in \cite{YIJIN,YIetalCDC}. Its basic idea is to translate system
state observation into the estimation of certain constant quantities. The interested reader may refer to 
\cite{YIetalTAC19} for the geometric interpretation to PEBOs. In contrast to the case with at least two non-collinear vectors in \cite{YIJIN,YIetalCDC}, in this paper we consider a more challenging scenario with only a single vector measurement available under a weak excitation condition. We are unaware of any previous results capable to deal with such a case. In the first observer design, after translating the problem into on-line parameter identification, we propose a mechanism to integrate both the historical and current information to achieve uniform convergence. The second proposed scheme uses a filter to generate a ``virtually'' measurable vector, which remains non-collinear with respect to the given reference vector.

{\em Notation.} $I_n \in \rea^{n\times n}$ represents the identity
matrix of dimension $n$, and  $0_n \in \rea^n$ and $0_{n\times m}\in
\rea^{n\times m}$ denote the zero column vector of dimension $n$ and the zero matrix of dimension $n\times m$, respectively. We use $\mathbb{N}$ to represent the set of all natural integers, and $\mathbb{N}_+$ for the set of positive integers. We also define the skew-symmetric matrix 
$$
\calj:= \begmat{0 & -1 \\ 1 & 0}.
$$
Given a square matrix $A\in \rea^{n\times n}$ and a vector $x\in \rea^n$, the Frobenius norm is defined as $\|A\| = \sqrt{\tr(A^\top A)}$, and $|x|$ represents the standard Euclidean norm. The $n$-sphere is
defined as $\mathbb{S}^n := \{x\in \rea^{n+1}: |x|=1\}$, and we use $SO(3)$ to
represent the special orthogonal group, and ${\mathfrak {so}}(3)$ is the
associated Lie algebra as the set of skew-symmetric matrices satisfying
$SO(3)=\{R\in \rea^{3\times3}|R^\top R = I_3, ~ \det(R) =1\}$. Given a variable $R\in SO(3)$, we use $|R|_I$ to represent the normalized distance to $I_3$ on $SO(3)$ with $|R|_I^2:= {1\over4}\tr(I_3 - R) $. The operator $\skew(\cdot)$ is defined as $\skew(A) := \hal (A- A^\top)$ for a square matrix $A$. Given $a \in \rea^3$, we define the operator $(\cdot)_\times$ as 
$$
a_\times := \begmat{ 0 & - a_3 & a_2 \\  a_3 & 0 & -a_1 \\ -a_2 & a_1 & 0 } \in
{\mathfrak {so}}(3),
$$
and its inverse operator is defined as $\vex(a_\times) = a$.

The paper is organized as follows. In Section \ref{sec2}, we recall the dynamical models and some basic definitions used in the paper. It is followed by the necessary and sufficient condition for observability in Section \ref{sec3}. Based on the proposed condition, we introduce two nonlinear attitude observer design in Section \ref{sec4}, the simulation results of which are presented in Section \ref{sec6}. Some discussions with practical considerations may be found in Section \ref{sec5}. The paper is wrapped up by some concluding remarks in Section \ref{sec7}. A preliminary version of this paper has been submitted to the 2022 IFAC Symposium on Nonlinear Control Systems (NOLCOS).

%%%%%%%%%%%%%%%%%%%%%
\section{Problem Formulation}
\label{sec2}
%%%%%%%%%%%%%%%%%%%%%

The aim of this note is to study the observability and observer design of the rotation matrix representing the coordinates of the body-fixed frame $\{B\}$ with respect to the coordinates of the inertial frame $\{I\}$, which lives in the group $SO(3)$. Its dynamics is given by
\begin{equation}
\label{dot:R}
\dot R = R \omega_\times, \quad R(0) = R_0
\end{equation}
with the rotational velocity $\omega \in \rea^3$ measured in the body-fixed
coordinate. Assume there is a vector $g\in \mathbb{S}^2$, known in the inertial
frame, being measured in the body-fixed frame, and the output is 
\begin{equation}
\label{yb}
 y_{\tt B} = R^\top g
\end{equation}
with $y_{\tt B} \in \mathbb{S}^2$, which is known as complementary measurement.
We also consider the compatible measurement $y_{\tt I}$, i.e., a known vector
$b \in \mathbb{S}^2$ in the body-fixed frame is measured in the inertial frame

\begequ
\label{yi}
y_{\tt I} = R b
\endequ
with $y_{\tt I} \in \mathbb{S}^2$. It is referred to \cite[Sec. II]{TRUetal}
for more details about  ``complementary'' and ``compatible'' measurements and their practical motivation. 

Before closing this section, let us recall some definitions used throughout the
paper.

\begin{definition}\label{def:distinguishability}\rm (\emph{Distinguishability \cite{BER}}) Consider an open
set $\mathcal{X}\subset\rea^n$ and a complete nonlinear system
\begin{equation}
\label{NL}
\dot x = f(x,t) , \quad y = h(x,t)
\end{equation}
with state $x\in\rea^n$ and output $y\in \rea^m$. The system \eqref{NL} is
distinguishable on $\cal X$ if for all $(x_a,x_b) \in \mathcal{X} \times
\mathcal{X}$,
$$
h\big(X(t; t_0,x_a),t\big)  =  h\big(X(t;t_0,x_b),t\big), ~\forall t \ge t_0 \implies 
x_a = x_b,
$$
in which $X(t;t_0,x_a)$ represents the solution at time $t$ of \eqref{NL}
through $x_0$ at time $t_0$. In this paper, we focus on the particular case
$t_0=0$.
\end{definition}

\begin{definition}\label{def1}\rm
(\emph{Persistent and interval excitation \cite{ORTetalReview}}) Given a
bounded signal $\phi:\rea_+ \to \rea^n$, it is persistently excited (PE) if
$$
     \int_{t}^{t+T} \phi(s)\phi^\top(s) ds \ge \delta I_n, ~ \forall t\ge 0
$$
    for some $T>0,\delta >0$; or intervally excited (IE), if there exists $T\ge
0$ such that  for some $\delta >0$
$$
    \int_{0}^{T} \phi(s)\phi^\top(s) ds \ge \delta I_n.
$$
\end{definition}

%
%%%%%%%%%%%%%%
\section{Necessary and Sufficient Conditions to Observability}
\label{sec3}
%%%%%%%%%%%%%%
%

First, we consider the observability for the case with multiple measurements 
\begequ
\label{yij}
\begin{aligned}
y_{{\tt B},i} & ~=~ R^\top g_{i}, \quad &i \in \ell_1 := \{ 1,\ldots,n_{1}\}
\\
y_{{\tt I},j} & ~=~ R b_{j}, \quad & j \in \ell_2:=\{ 1,\ldots, n_2\}
\end{aligned}
\endequ
with $n_1,n_2 \in \mathbb{N}$.\footnote{If $n_i=0$ ($i=1,2$), then the set $\ell_i$ is defined as the empty set $\emptyset$.} It is clear that the single measurement is corresponding to the case $n_1+n_2 =1$, for which we will construct two asymptotically convergent observers in the next section.

In the following proposition, we uncover a necessary and sufficient condition to the distinguishability for attitude estimation. 

\begin{proposition}
\rm\label{prop:cond}
The time-varying system \eqref{dot:R} with the output \eqref{yij}, and $n := n_1+n_2 \ge 1$, is distinguishable if and only if there exist two moments $t_1, t_2 \ge 0$ such
that
\begin{equation}
\label{cond:1}
\begin{aligned}
  \sum_{\substack{i,l \in \ell_1,  j,k \in \ell_2}} 
  \Big| g_i(t_1)\times g_l(t_2)\Big| 
  & ~+~
  \Big|g_i(t_1)_\times R_0\Phi(0,t_2) b_j(t_2)  \Big|
  \\
  & ~ +~ 
  \Big| b_j(t_1)_\times \Phi(t_1,t_2) b_k(t_2) \Big| >0,
\end{aligned}
\end{equation}
in which $\Phi(s,t)$ is the state transition matrix of the time-varying
system matrix $\omega_\times(t)$ from $s$ to $t$.
\end{proposition}
%=======
\begin{proof}
The state transition matrix $\Phi(s,t)$ of the linear time-varying (LTV) system
$$
\dot x = \omega_\times x 
$$
with $x\in \rea^3$ is defined as
$$
\begin{aligned}
{\partial \over \partial t}\Phi(s,t) & ~=~ \omega_\times(t)\Phi(s,t)
\\
\Phi(s,s) & ~=~ I_3.
\end{aligned}
$$
It is equivalent to define $\Phi(s,t) = Q(s)^{-1} Q(t)$, in which $Q \in SO(3)$
is generated by the dynamics
\begin{equation}
\label{dot:Q}
\dot Q = Q \omega_\times, \quad Q(0) = I_3
\end{equation}
with $Q\in SO(3)$. From
$$
\dot{\aoverbrace[L1R]{R Q^\top}} ~= \dot R Q^\top - R Q^\top \dot Q Q^\top = 0,
$$
we have for all $t,s\ge 0$
$$
\begin{aligned}
R(t) Q(t)^\top = R(0) Q(0)^\top
& ~ \iff ~
R(t) ~=~ R_0 Q(t)
\\
& ~ \iff ~
R(t) ~=~ R(s) \Phi(s,t),
\end{aligned}
$$
with $R_0 := R(0)$.

Now we collect all the measured outputs in the vector 
$$
\bar y = \col(y_{{\tt B},1}, \ldots, y_{{\tt B},n_1}, y_{{\tt I},1}, \ldots,
y_{{\tt I},n_2}). 
$$
With a slight abuse of notation, we denote the output signal $\bar y$ from the
initial condition $R_0\in SO(3)$ as $\bar y(t; R_0)$. In terms of Definition \ref{def:distinguishability}, the system is distinguishable from $t=0$
if and only if 
\begequ
\label{impl:1}
\bar y (t; R_a) \not \equiv \bar y (t; R_b) 
\quad \implies \quad
R_a \neq R_b
\endequ
for any $R_a, R_b \in SO(3)$. Clearly, the above condition \eqref{impl:1} is
equivalent to the \emph{identifiability} of the constant matrix $R_0 \in SO(3)$ from the
nonlinear regressor equation
\begequ
\label{NLRE1}
\bar y = h(R_0, t)
\endequ
with the equation
$$
h(R_0,t) := \begmat{Q^\top (t) R_0^\top g_1(t)\\ \vdots \\ Q^\top(t) R_0^\top
g_{n_1}(t)
\\
R_0 Q(t) b_1(t) \\ \vdots \\ R_0 Q(t) b_{n_2}(t)
}.
$$
The regressor equation \eqref{NLRE1} can be equivalently rewritten as
\begequ
\label{NLRE2}
Y(t) =  R_0^\top \phi(t), \quad R_0\in SO(3)
\endequ
with $Y\in \rea^{3\times n}$ and $\phi \in \rea^{3\times n}$ given by
$$
\begin{aligned}
Y & ~:=~  Q\begmat{y_{{\tt B},1},\ldots, y_{{\tt B},n_1}, b_1,\ldots, b_{n_2}}
\\
\phi & ~:=~  \begmat{g_1, \ldots, g_{n_1}, y_{{\tt I},1}, \ldots, y_{{\tt
I},n_2}}.
\end{aligned}
$$
Hence, the identifiability of the constant matrix $R_0$ on $SO(3)$ from the nonlinear regression model \eqref{NLRE1} is equivalent to the solvability of the Wahba problem for the regression model \eqref{NLRE2} over time \cite{WAH} -- invoking that \eqref{NLRE2} holds for all $t \ge 0$. That is the existence
of moments $t_1,t_2\ge 0$ such that
\begequ
\label{cond:phi}
\phi_i(t_1) \times \phi_j(t_2)  \neq 0
\endequ
for some $i,j \in \{1,\ldots, n\}$, with $\phi_i$ representing the $i$-th
column vector of $\phi$.

The last step of the proof is to show that the condition \eqref{cond:phi} is
equivalent to \eqref{cond:1}. There are totally three possible cases when
\eqref{cond:phi} holds true: 1) $i,j \in \{1,\ldots, n_1\}$, 2) $i,j \in \{n_1+1,\ldots,
n\}$, and $i\in \{1,\ldots, n_1\}, j\in \{n_1+1,\ldots, n\}$.\footnote{We do
not distinguish the order of $i$ and $j$.} For the first case, the condition
\eqref{cond:phi} is equivalent to
\begequ
\label{case:1}
\sum_{i, l \in \ell_1} \Big| g_i(t_1) \times g_l (t_2)\Big|>0.
\endequ
The second case is equivalent to for some $j, k \in \ell_2$
\begequ
\label{case:2-1}
\begin{aligned}
 & \quad y_{{\tt I},j}(t_1) \times y_{{\tt I},k} (t_2) \neq 0
 \\
 \iff & \quad [R(t_1) b_j(t_1)]_\times R(t_2) b_k(t_2) \neq 0
 \\
 \iff & \quad R(t_1) [ b_j(t_1)]_\times R(t_1)^\top R(t_2) b_k(t_2) \neq 0
 \\
  \iff & \quad [ b_j(t_1)]_\times R(t_1)^\top R(t_2) b_k(t_2) \neq 0
  \\
   \iff & \quad [ b_j(t_1)]_\times \Phi(t_1,t_2) b_k(t_2) \neq 0
\end{aligned}
\endequ
where in the second implication we use the identity $(Rb)_\times = R b_\times
R^\top$, the full-rankness of $R(t_1)$ in the third implication, and in the
last
$$
\begin{aligned}
R(t_1)^\top R(t_2) & ~=~ Q(t_1)^\top R_0^\top R_0 Q(t_2) 
\\
& ~=~ \Phi(t_1, t_2).
\end{aligned}
$$
Note that the last line of the condition \eqref{case:2-1} can be compactly written as
\begin{equation}
\label{case:2}
\sum_{j,k \in \ell_2} \Big| b_j(t_1)_\times \Phi(t_1,t_2) b_k(t_2) \Big| >0.
\end{equation}
Similarly, we get that for the third case the condition \eqref{cond:phi} is
equivalent to
\begin{equation}
\label{case:3}
\sum_{i\in \ell_1, ~j \in \ell_2} \Big|g_i(t_1)_\times R_0\Phi(0,t_2) b_j(t_2) 
\Big|>0.
\end{equation}
Combining these three cases, it is sufficient to obtain \eqref{cond:1}. On the other hand, since each term in \eqref{cond:1} is non-negative, if the condition \eqref{cond:1} holds, at least one of the above cases should be satisfied. We complete the proof.
\end{proof}

For the case with \emph{only} complementary or compatible measurements
$(n_1\cdot n_2 =0)$, then the distinguishability condition becomes 
$$
  \sum_{\substack{i,l \in \ell_1,  j,k \in \ell_2}} 
  \Big| g_i(t_1)\times g_l(t_2)\Big| 
   +
  \Big| b_j(t_1)_\times \Phi(t_1,t_2) b_k(t_2) \Big| >0.
$$
If there are two types of measurements, the identifiability is dependent of the
initial attitude $R_0$, and this implies that some region in $SO(3)$ may be not
distinguishable for a given specific trajectory. However, the following corollary shows that such a region has zero Lebesgue measure in the group $SO(3)$. Note that the condition below does not rely on the initial attitude $R_0$.

\begin{corollary}\rm
\label{cor:modified}
If the condition \eqref{cond:1} is replaced by the initial
attitude-independent term
\begin{equation}
\label{cond:modified}
\begin{aligned}
  \sum_{\substack{i,l \in \ell_1,  j,k \in \ell_2}} 
  \Big| g_i(t_1)\times g_l(t_2)\Big| 
  & ~+~
  \Big|g_i(t_1)_\times \Phi(0,t_2) b_j(t_2)  \Big|
  \\
  & ~ +~ 
  \Big| b_j(t_1)_\times \Phi(t_1,t_2) b_k(t_2) \Big| >0,
\end{aligned}
\end{equation}
the distinguishability is guaranteed almost surely.\footnote{We refer to that the set of initial attitudes making the system lose distinguishability has zero Lebesgue measure in the entire state space.}
\end{corollary}
\begin{proof}
It is given in the appendix.
\end{proof}

\begin{remark}
\rm 
In \cite{TRUetal} the authors propose the following
\emph{sufficient} (but not necessary, c.f. \cite[Remark 3.9]{TRUetal}) condition for distinguishability of the
given system.
\begin{equation}
\label{cond:trumpf}
\begin{aligned}
&\lambda_2 \left( \sum_{i\in \ell_1}\int_0^T g_i(s) g_i^\top(s) ds \right)
\\
& \hspace{1.7cm} +
\left\|
\int_0^T \sum_{j\in \ell_2} \left(\omega_\times b_j(s) + {d\over ds}b_j(s)
\right)ds
\right\| >0,
\end{aligned}
\end{equation}
for some $T>0$, with $\lambda_2(\cdot)$ representing the second largest
eigenvalue of a square matrix. Note that in the above condition it is necessary to
impose (piece-wise) smoothness of the signals $b_j$. In the following corollary, we show that the above condition is sufficient to the proposed necessary and sufficient condition \eqref{cond:1}. 
\end{remark}

\begin{corollary}
\label{cor:subset}\rm
Consider the time-varying system \eqref{dot:R} with the output \eqref{yij}, and $n := n_1+n_2 \ge 1$. If \eqref{cond:trumpf} holds, then the condition in
Proposition \ref{prop:cond} is also verified.
\end{corollary}
\begin{proof}
It is given in the appendix.
\end{proof}

%
%%%%%%%%%%%%%%%%%%
\section{Attitude Observer for A Single Vector Measurement}
\label{sec4}
%%%%%%%%%%%%%%%%%%
%

In this section, we show that the distinguishability condition -- identified in Proposition \ref{prop:cond} -- is adequate to design a continuous-time observer with almost globally asymptotically convergent estimate to the unknown attitude. %We focus on the case with only a single vector measurement, which is more challenging than the multiple vector case.

Since the scenario with only a single vector measurement is more challenging than the multiple vector case, we focus on the former in this section. The main results can be extended to the case with multiple vector measurements in a straightforward manner.

\subsection{Attitude Observer Using Integral Correction Term}
Let us consider the observer design with a single complementary measurement \eqref{yb}. In the first observer design, we construct a dynamic extension -- following the PEBO methodology \cite{ORTetalscl} -- in order to reformulate attitude estimation as an on-line consistent parameter identification problem. By adding an elaborated construction of ``integral''-type correction term, we are able to achieve asymptotic stability of the observer.

%We have the following.

\begin{proposition}\rm
\label{prop:obs1}
For the system \eqref{dot:R} with the complementary output \eqref{yb}, we assume that all signals are piece-wisely continuous and the reference satisfies the distinguishability
condition, i.e.,  
\begequ
\label{cond:g}
\exists t_1, t_2 >0, \quad \big|g(t_1)\times g(t_2)\big| >0,
\endequ
with a \emph{known} bound $T>0$ on the distinguishability
interval.\footnote{Namely, there exists a known $T>0$ such that $0\le t_1 <
t_2 \le T$.} The attitude observer
\begin{equation}
\label{dot:Q}
 \dot Q  = Q\omega_\times
\end{equation}
with $Q(0)\in SO(3)$ and
\begin{equation}
\label{dot:Qc}
\begin{aligned}
\dot {\hat Q}_c  = \eta_\times \hat{Q}_c 
, \quad
 \hat R  = \hat Q_c^\top Q
\end{aligned}
\end{equation}
with
$$
\begin{aligned}
\eta & = \gamma_{\tt P} (\hat{Q}_c g) \times (Qy_{\tt B}) + \gamma_{\tt I} \xi
\\
\xi & = 2\vex \big( \skew(A \hat Q_c^\top) \big)
\\
\dot A & =\left\{ \begin{aligned}
&Q y_{\tt B} g^\top, && t\in [0, T) \\ & 0_{3\times 3}, && t \ge T.
\end{aligned}
\right.
\end{aligned}
$$
with the gains $\gamma_{\tt P}, \gamma_{\tt I}>0$ and $A(0)= 0_{3\times 3}$,
guarantees $\hat R(t) \in SO(3)$ for all $t\ge 0$ and the convergence
\begequ
\label{convergence:hat_R}
\lim_{t\to\infty} \| \hat R(t) - R(t)  \| =0
\endequ
almost globally. 
\end{proposition}
%=====
\begin{proof}
Let us consider the dynamic extension
$
\dot Q = Q\omega_\times
$
for the initial condition $Q(0)\in SO(3)$. By defining a variable $E(R,Q) =
QR^\top$ -- which also lives in $SO(3)$ -- we have
$$
\dot E = \dot Q R^\top - Q R^\top \dot R R^\top = 0.
$$
Therefore, there exists a constant matrix $Q_c \in SO(3)$ such that
\begin{equation}
\label{id:QR}
Q(t)R^\top(t) = Q_c, \quad \forall t\in [0,+\infty).
\end{equation}
Note that $Q(t)$ is an available signal by construction, and $Q_c$ is unknown. Invoking \eqref{id:QR} and the full-rankness of $Q$, the estimation of $R$ is equivalent to the one of $Q_c$.

Based on the above idea, we construct the following auxiliary system
\begin{equation}
\label{aux:2}
\Sigma_c: \left\{~
\begin{aligned}
\dot Q_c & = Q_c (\omega_c)_\times
\\
y_c & = Q_c b_c,
\end{aligned}
\right.
\end{equation}
in which $Q_c\in SO(3)$ is constant thus $\omega_c = 0_3$, the output 
$$
y_c(t) ~:=~ Q(t) y_{\tt B}(t),
$$
and the ``body-fixed coordinate'' reference $b_c:= g$. It is clear
that the system $\Sigma_c$ is exactly in the same form as the kinematic model with a compatible measurement \eqref{dot:R} and
\eqref{yi}.

We now define the estimation error of $\tilde Q_c:= \hat Q_c Q_c^\top$, the
dynamics of which is given by
\begequ
\lab{dot:tQc}
\dot{\tilde Q}_c = \dot{\hat Q}_c Q_c^\top - \hat Q_c Q_c^\top \dot Q_c
Q_c^\top = \eta_\times \tilde{Q}_c.
\endequ
The term $\eta$ satisfies  
\begin{equation}
\begin{aligned}
\eta_\times & ~=~
\gamma_{\tt P} [(\hat Q_c g)\times (Qy_{\tt B})]_\times + \gamma_{\tt I}
\xi_\times
\\
& ~= ~
\gamma_{\tt P} \big[ Qy_{\tt B} (\hat Q_c g)^\top - \hat Q_c g (Qy_{\tt
B})^\top \big]    + \gamma_{\tt I} \xi_\times
\\
& ~=~ \gamma_{\tt P} \big( y_cy_c^\top \tilde Q_c^\top - \tilde Q_c y_c y_c^\top
\big)    + \gamma_{\tt I} \xi_\times
\end{aligned}
\end{equation}
in which for $t\in [0,T]$
$$
\begin{aligned}
\xi_\times & ~=~  \int_0^t   
\Big[
 Q(s)y_{\tt B}(s) \big(\hat Q_c(t) g(s)\big)^\top - \hat Q_c(t)
g(s) \big(Q(s)y_{\tt B}(s)\big)^\top 
\Big]ds
\\
& ~=~ 
2\skew\left(\int_0^t y_c(s)y_c^\top(s) ds \cdot \tilde Q_c^\top\right),
\end{aligned}
$$
and for $t>T$ we have $\xi(t) = \xi(T)$.

Consider the candidate Lyapunov function $V(\tilde Q_c) = 3 - \tr(\tilde Q_c)$,
which has its minimal value zero has at $\tilde Q_c = I_3$. It yields 
$$
\begin{aligned}
\dot V & = - \tr(\eta_\times \tilde Q_c) 
\\
& = - \gamma_{\tt P} \tr \Big( y_c y_c^\top -\tilde Q_c  y_c y_c^\top \tilde Q_c \Big)
\\
& \hspace{0.4cm} 
-  \gamma_{\tt I}
 \int_0^t \tr \Big( y_c(s) y_c^\top(s)  -  \tilde Q_c 
 y_c(s) y_c^\top(s) \tilde Q_c \Big) ds
\\
& = - \gamma_{\tt P} y_c^\top (I- \tilde Q_c^2) y_c
-  \gamma_{\tt I} \int_0^t y_c^\top(s) (I- \tilde Q_c^2) y_c(s)ds
\\
& = - 2 \vex\big(\skew(\tilde Q_c)  \big)^\top \Gamma \vex\big(\skew(\tilde
Q_c)  \big)
\\
& \le -\lambda_{\tt min} (\Gamma ) \|\skew(\tilde Q_c)\|^2,
\end{aligned}
$$
where in the fourth equation we have used $2|v|^2 = \|v_\times\|^2$ for any
$v\in \rea^3$, with the definition of $\Gamma$ as
\begin{equation}
\label{Gamma}
\Gamma = \Gamma_{\tt P} + \Gamma_{\tt I}
\end{equation}
with
$$
\begin{aligned}
\Gamma_{\tt P}(t) & ~:=~ \gamma_{\tt P} (I- y_c(t) y_c^\top (t))
\\
\Gamma_{\tt I}(t) & ~:=~ \left\{
\begin{aligned}
&\gamma_{\tt I} \int_0^t \Big(I- y_c(s) y_c^\top (s) \Big)ds, && t\in [0,T]
\\
&\Gamma_{\tt I}(T), && t >T.
\end{aligned}
\right.
\end{aligned}
$$

Let us study the property of the matrix $\Gamma \in \rea^{3\times 3}$. From the
assumption $|g(t_1) \times g_2(t_2)|>0$ for some $t_1,t_2\le T$, we have
$$
\begin{aligned}
|y_c(t_1) \times y_c(t_2)| & ~=~ |(Q_cg(t_1))_\times (Q_c g(t_2))|
\\
& ~=~ |Q_c g(t_1) \times g(t_2)|
\\
& ~>~ 0.
\end{aligned}
$$
It implies that
\begequ
\label{cond:sum1}
 2I - y_c(t_1)y_c^\top(t_1) - y_c(t_2) y_c^\top (t_2)  >0,
\endequ
in which we have used the fact that for $a,b\in \mathbb{S}^2$, $|a\times b|>0$
implies the positiveness of $(2I-aa^\top - bb^\top)$; see \cite[Lemma
A.2]{TRUetal} for its proof. On the other hand, in terms of the continuity of
$y_c$ and \eqref{cond:sum1}, we have
$$
\int_{t_1}^{t_1+ \varepsilon} I - y_c(s)y_c^\top(s)ds + \int_{t_2}^{t_2+
\varepsilon} I - y_c(s)y_c^\top(s)ds > 0
$$
for some sufficiently small $\varepsilon>0$, and thus
$$
\int_0^T I - y_c(s)y_c^\top(s)ds > 0 \quad \implies \quad \lambda_{\tt
min}(\Gamma(t))>0, \; \forall t\ge T.
$$
From $\dot V \le -\lambda_{\tt min}(\Gamma)\|\skew(\tilde{Q}_c(s))\|^2$, we get
that $0\le V(\tilde Q_c(t)) \le V(\tilde Q_c(0))$ and
$$
V(\tilde Q_c(t)) - V(\tilde Q_c(0)) \le  - \int_0^t \lambda_{\tt
min}(\Gamma(s))\|\skew(\tilde{Q}_c(s))\|^2 ds.
$$
By taking $t\to\infty$, in terms of Barbalat's lemma and the boundedness of the time derivative of $\tilde Q_c$, we have $\|\skew(\tilde Q_c)\| \to 0$. The set
$\{\tilde Q_c \in SO(3) : \|\skew(\tilde Q_c)\| =0\}$ only contains the stable
equilibrium $I_3$ and the unstable equilibria $U^\top \diag(1,-1,1)U$. For the
latter case, the Lyapunov $V(\tilde Q_c)$ equals to its maximal value, thus
having zero Lebesgue measure. Therefore, the dynamics \eqref{dot:tQc} is almost
globally asymptotically stable on $SO(3)$. Invoking the algebraic relation $R =
Q_c^\top Q$, we complete the proof.
\end{proof}

\begin{remark}\rm
In the above attitude observer design, the error term $\eta$ contains two parts
$$
\eta ~ =~\underbrace { \small \gamma_{\tt P}(\hat{Q}_c g) \times
(Qy)}_{\mbox{current}} ~~ +  \underbrace{\gamma_{\tt I}\xi}_{\small
\mbox{historical}},
$$
which may be viewed as an observer design using a ``proportional + integral''-type error term. The first term only utilizes the current information, making it behave as an on-line design. The second ``integral'' term, which may be written as
$$
\begin{aligned}
\xi(t)  & = 
\left\{\begin{aligned}
&\int_0^t  [\hat{Q}_c(t) g(s)] \times [Q(s)y_{\tt I}(s)]ds, & &t \in [0, T]
\\
&~\xi(T), & & t\ge T
\end{aligned}
\right.
\end{aligned}
$$
enables to achieve asymptotic convergence of the estimation error under the
extremely weak identifiability condition \eqref{cond:g}. The gain parameters $\gamma_{\tt P}$ and $\gamma_{\tt I}$ can be used as the weights on how we trust the current and historical data.
\end{remark}

\begin{remark}\rm
The bound $T>0$ is used in the dynamics of the variable $A$ in order to be able
to guarantee its boundedness. Indeed, the bound $T$ is not necessarily known
apriori, since the distinguishability condition \eqref{cond:g} is an
easily-checkable condition on measured quantities. The proposed scheme may be
modified as an \emph{adaptive} design in which such a condition is checked
online continuously, and the dynamics of $A$ simply changes until the condition
holds. It is also natural to replace the condition \eqref{cond:g} by $|g(t_1)
\times g(t_2)|>\delta$ for some $\delta>0$, to deal with sensor noise.
\end{remark}

\begin{remark}\rm
As is shown above, the ``integral'' term only accumulates information in the interval $[0,T]$, which, however, does not have the sort of ``fading memory'' property on past measurements. As long as the excitation condition, which is easily monitored on-line, the observer performance can be improved considering the moving interval $[t-T, t]$ rather than $[0,T]$ in Proposition \ref{prop:obs1}.
\end{remark}

\subsection{Attitude Observer Using Virtual Vectors}

In this subsection, we provide an alternative observer design, which does not
need the information of $T$. The basic idea is to generate a new
``\emph{virtual}'' vector measurement 
\begequ
\label{yvqc}
y_v = Q_c b_v
\endequ
from the real measurement \eqref{aux:2}, such that 
\begequ
\label{bvbc}
b_v \times b_c \neq 0
\endequ
uniformly after some moment with $b_c = g$. Then, it becomes the well-studied attitude observer design problem with (not less than) two non-collinear vectors, which has been well addressed in the literature. 

\begin{proposition}\rm
\label{prop:obs2}
For the system \eqref{dot:R} with the output \eqref{yb}, we assume that all
signals are continuous and satisfy $\exists t_i >0$ ($i=1,2,3$)
\begequ
\label{cond:2}
\det \Big( \begmat{g(t_1) & g(t_2) & g(t_3)} \Big) \neq 0
 %\big|g(t_i)\times g(t_j)\big| >0, \quad \forall i,j  \in \{1,2,3\},\; i\neq j.
\endequ
Consider the dynamic extension \eqref{dot:Q} and the LTV filter
\begin{equation}
\label{filter:ltv}
\begin{aligned}
\dot Z & ~= ~ \gamma_z g(y_{\tt B}^\top Q^\top - g^\top Z)
\\
\dot \Omega & ~=~ -\gamma_z g g^\top \Omega, ~\Omega(0) = I_3
\\
\dot P & ~ = ~ \gamma (I_3 - \Omega)^\top\left[ Z - \Omega Z_0 - (I_3 - \Omega)
P\right]
\end{aligned}
\end{equation}
with $\gamma, \gamma_z>0$, $Z_0:= Z(0)$, and the filtering outputs 
$$
b_v = U g, \quad y_v = {P^\top Ug }
$$
in which $U:= \diag(\calj,1) \diag(1, \calj)$. Then, the observer
\eqref{dot:Qc} with
\begequ
\label{eta:2}
\eta ~=~ \gamma_c (\hat Q_c g) \times (Qy_{\tt B}) + \gamma_v (\hat Q_c b_v)
\times y_v
\endequ
and the gains $\gamma_c, \gamma_v>0$ guarantees $\hat R(t)\in SO(3)$ for all
$t\in [0,\infty)$ and the convergence \eqref{convergence:hat_R} almost globally.
\end{proposition}
%%%%
\begin{proof}
First, let us study the property of the LTV filter \eqref{filter:ltv}. Note
that $U$ can be decomposed as the product of three basic (element) rotations $U
= R_x(\theta_1) R_y(\theta_2) R_z(\theta_3)$ with $\theta_1=\theta_3 = {\pi
\over 2}$ and $\theta_2 =0$. Hence, $g \times (Ug) \neq 0$, verifying the
equation \eqref{bvbc}. Then, we need to verify \eqref{yvqc} in an asymptotic
sense.

According to the proof in Proposition \ref{prop:obs1}, we may reformulate the
estimation of $R$ as the one of $Q_c$.
From the dynamics of $\dot Z$, one has
$$
\begin{aligned}
{d\over dt}(Z - Q_c^\top) & ~=~ \gamma_z g (y_c^\top -g^\top Z)
\\
& ~=~ - \gamma_z g g^\top (Z- Q_c^\top),
\end{aligned}
$$
and thus
$$
Z - Q_c^\top  ~=~ \Omega(Z_0- Q_c^\top)
~ \implies ~
Z - \Omega Z_0 ~=~(I-\Omega) Q_c^\top.
$$
Then, it yields 
\begequ
\label{dot:PQc}
\begin{aligned}
{d\over dt} (P - Q_c^\top) & ~=~  -\gamma \phi\phi^\top(P - Q_c^\top).
\end{aligned}
\endequ
with $\phi:= I_3 - \Omega^\top$.

From the condition \eqref{cond:2}, for \emph{any} non-zero $x \in \rea^3$ it may
always be represented as
\begequ
x = c_1 g(t_1) + c_2 g(t_2) + c_3g(t_3),
\endequ
with at least one of $c_i$ ($i=1,2,3$) non-zero. Hence
$$
\begin{aligned}
x^\top \sum_{i =1}^3  g(t_i) g(t_i)^\top x > 0 & ~\implies ~
\sum_{i =1}^3  g(t_i) g(t_i)^\top >0
\\
& ~\implies ~ \int_{t_i}^{t_i + \varepsilon}
\sum_{i =1}^3  g(s) g(s)^\top ds >0
\\
& ~ \implies ~
\int_0^T g(s)g(s)^\top ds >0.
\end{aligned}
$$
for some $T > t_i$ and some sufficiently small $\varepsilon>0$. It implies that
the vector signal $g$ is intervally excited. As a result, the matrix $\phi =
I_3 - \Omega^\top$ is PE \cite[Proposition 2]{WANetal}. Following \cite[Thm.
2.5.1]{SASBOD}, we are able to derive 
\begequ
\label{P2Qc}
\lim_{t\to\infty}\|P - Q_c^\top\| = 0 
\quad \implies \quad
\lim_{t\to\infty } |P^\top Ug - Q_c g| =0
\endequ
exponentially fast. Hence, the algebraic equation \eqref{yvqc} is guaranteed
asymptotically, i.e., 
$$
y_v = Q_c b_v +\et
$$ with an exponentially decaying term $\et$.

The last step is to study the convergence of the estimation error of $Q_c$, which
is defined as $\tilde Q_c :=  Q_c \hat Q_c^\top$. The observer $\dot {\hat
Q}_c  = \eta_\times \hat{Q}_c $ with the output error term \eqref{eta:2} may be
written as the ``standard'' compatible observer 
$$
\dot{\hat Q}_c = \hat Q_c \bigg[ \hat Q_c \Big( \gamma_c(\hat Q_c b_c)\times
b_c +\gamma_v(\hat Q_c b_v + \et) \times b_v  \Big) \bigg]_\times
$$
for the auxiliary dynamics \eqref{aux:2} with $\omega_c = 0_3$. If the term
$\et$ was zero, in terms of \cite[Thm 4.3]{TRUetal} and the uniform
non-collinear relation $b_c\times b_v \neq 0$, we would obtain $\hat Q_c \to
Q_c$ as $t\to \infty$ almost globally. However, the term $\et$ is exponentially
decaying to zero, and we may follow the perturbation analysis in
\cite[Proposition 6]{YIJIN} using a time-varying Lyapunov function to obtain
the almost global asymptotic stability. We omit the details here. Invoking the
algebraic relation \eqref{id:QR}, we complete the proof.
\end{proof}

\begin{remark}\rm
The condition \eqref{cond:2} with three non-collinear moments is slightly
stronger than the one \eqref{cond:g}. However, it removes the necessity of having a
priori known bound $T>0$ for the observer design in Proposition
\ref{prop:obs2}. The key idea to construct a new ``virtual'' non-collinear reference vector $b_v$ is similar to some recent results on generation of PE regressors from those only satisfying IE \cite{WANetal,YIJIN,BOBetal}. We refer the reader to the monograph \cite{BES} for a discussion of excitation in observer design. 
\end{remark}

\begin{remark}\rm
For the case with a single compatible measurement \eqref{yi}, we may still get the auxiliary model \eqref{aux:2} by designing the dynamic extension $\dot Q = Q\omega_\times$, but with the new definitions of $y_c:= y_{\tt I}$ and $g_c: = Qb$. Then, the above two designs are capable to solve the problem with slight modifications accordingly. 
\end{remark}

%
%%%%%%%%%%%%%%%%%%%%%
\section{Discussions}
\label{sec5}
%%%%%%%%%%%%%%%%%%%%%
%

In this section, we show some practical issues in attitude estimation --
intermittent and delayed measurements \cite{BERTAYauto19,BAHNAM} -- can be
easily (and even trivially) tackled by the proposed methodology.

\begin{remark}\rm (\emph{Delayed measurement})
Time delay in attitude estimation is generally unavoidable, which is usually
caused by low-quality data sampling and poor sensors \cite{BAHNAM}. A common
scenario is that a known delay $\tau$ appears in the single vector measurement,
i.e.
$
y(t) = y_{\tt I}(t-\tau) = R^\top(t-\tau) g(t-\tau),
$
in which $\tau(t)$ may be constant or time-varying. After designing the dynamic
extension \eqref{dot:Q}, the delayed output can be rewritten as
$
y(t) = Q^\top(t-\tau) Q_c g(t-\tau).
$
Then, we are still able to get the auxiliary model \eqref{aux:2} with
\begequ
\label{ycbc2}
y_c:= Q(t-\tau)y(t) , \quad b_c:= g(t-\tau).
\endequ
Since $Q_c$ is a \emph{constant} matrix on the special orthogonal group, the
observers in Propositions \ref{prop:obs1}-\ref{prop:obs2} can provide
asymptotically convergent attitude estimation by modifying the ``reference
vector'' $b_c$ and the vector $y_c$ as \eqref{ycbc2}.
\end{remark}

%%%%%%%%%%%%%%%%%%%%%%

\begin{remark}\rm (\emph{Intermittent measurement})
Some types of sensors only provide intermittent measurement $y_{\tt I}$ at some
instants of time $t_k$ ($k \in \mathbb{N}_+$). Let sequence $\{t_k\}_{k \in
\mathbb{N}_+}$ be strictly increasing, and $|t_1|$ and $|t_{k+1} - t_{k}|$ ($k\ge
1$) are upper and lower bounded by two positive constants. In the proposed
attitude PEBO framework, we have translated the estimation of the variable
$R(t)$ into that of constant $Q_c$, and thus the intermittent measurement
does not bring any difficulty in observer design. By defining the
vectors
\begin{equation}
\label{ycbc3}
    y_c(t) := Q(t_i) y(t_i), \quad b_c(t) := g(t_i), \quad  \forall t\in
[t_i,t_{i+1}].
\end{equation}
Again, we get the auxiliary model \eqref{aux:2} with the modified reference
vectors in \eqref{ycbc3}. The proposed two \emph{continuous-time} observers can
solve attitude determination.
\end{remark}

%%%%%%%%%%%%%%%%%%%%%%

\begin{remark}\rm
(\emph{Robustness}) In practice, the measurement vector is perturbed by sensor noise, i.e., $y_{\tt M} = y_{\tt B} + n_y$ with a bounded term $n_y$, and $y_{\tt M}, y_{\tt B} \in \mathbb{S}^2$. Then, the correction term $\eta$ of the observer in Proposition \ref{prop:obs1} becomes $\eta_{\tt M} := \eta + \Delta_\eta$, in which $\eta$ is the nominal part defined in Proposition \ref{prop:obs1} and $\Delta_\eta$ is the additive term stemmed from the noise term $n_y$. Since $\hat Q_c, Q, g$ and $y_{\tt B}$ all live in some compact sets, as well as the variable $A$ being the integral over a \emph{finite} interval $[0,T]$, there exists a constant $k>0$ such that $|\Delta_\eta(t)|\le k\|n_y\|_\infty$. For this case, the time derivative of the Lyapunov function becomes $\dot V \le -\lambda_{\tt min}(\Gamma)\|\skew(\tilde Q_c)\|^2 + |\tr((\Delta_\eta)_\times \tilde Q_c)|$. Hence, we are able to establish the robustness of the observer in the bounded-input-bounded-output (BIBO) sense. The is also verified via noisy simulations in the coming section. 
\end{remark}

\section{Simulations}
\label{sec6}
%%%%%%%%%%%%%%%%%%%%%
%

%Some simulations have been carried out with to evaluate the performance of the proposed observers. 

\begin{example}
\label{example:1}
We consider a single time-varying inertial vector 
\begequ
\label{g:simulation}
g(t) = \left\{
\begin{aligned}
& e_1, && t\in [0,5) \mbox{s}
\\ 
& e_3, && t \ge 5 \mbox{s},
\end{aligned}
\right.
\endequ
in which $e_i$ represents the $i$-th standard Euclidean basis in $\rea^3$.
Clearly, it satisfies the sufficient excitation condition \eqref{cond:g}, but
not for the persistently non-constant reference vector assumption in many works
\cite{BATetal}. The attitude of the rigid-body starts from the initial condition
$R(0)=\diag(-1,-1,1)$ under the rotational velocity $\omega = [0.23 ~ -0.5 ~
0.15]^\top$. We added noise to both the angular velocity readings and the vector measurements.

First, we evaluate the performance of the scheme in Proposition
\ref{prop:obs1}. The observer is initialized from $Q(0)=\hat Q_c(0)= I_3$,
with the gains $\gamma_{\tt P} =3$ and $\gamma_{\tt I}=1$. It
corresponds to the initial yaw, pitch and roll estimates all being $0^\circ$.
The results of simulations are shown in Fig. \ref{fig:1} in the form of Euler angles, and also see the norm of the estimation error $|\tilde R|_I$ in Fig.
\ref{fig:3}, which is drawn in a logarithmic scale for the $y$-axis. During
$[0,5]$ s, the error $\tilde R$ is converging to some non-zero constant under a
constant vector measurement. This is because a single vector output makes two
of three Euler angles partially observable \cite{MARSAR}. After 5s the model
satisfies the distinguishability, and then all Euler angles converge to their
true values. Note that the proposed scheme is robust \emph{vis-\`a-vis}
measurement noise. Then, we test the second observer design in Proposition
\ref{prop:obs2}, with the simulation results presented in the same figure. Though the reference vector $g$ in \eqref{g:simulation} does not satisfy the sufficient condition \eqref{cond:2}, it is interesting to observe that all the Euler angles converge to zero asymptotically for the same reference vector $g$. This implies that the condition \eqref{cond:2} is not necessary for the convergence of the second observer design. At the end, let us compare the proposed schemes to the complementary attitude observer in \cite{TRUetal}, whose convergence is guaranteed by a persistent excitation condition. Clearly, this is not satisfied by the inertial reference vector in \eqref{g:simulation}. We show the simulation results for in Fig. \ref{fig:3}. As expected, the estimate $\hat R$ from the observer in \cite{TRUetal} fails to converge to its true values. Besides, we note that the first design in Proposition \ref{prop:obs1} is less sensitive to measurement noise.

%%%%

\begin{figure*}[htbp!]
\centering
\includegraphics[width = 0.99\linewidth]{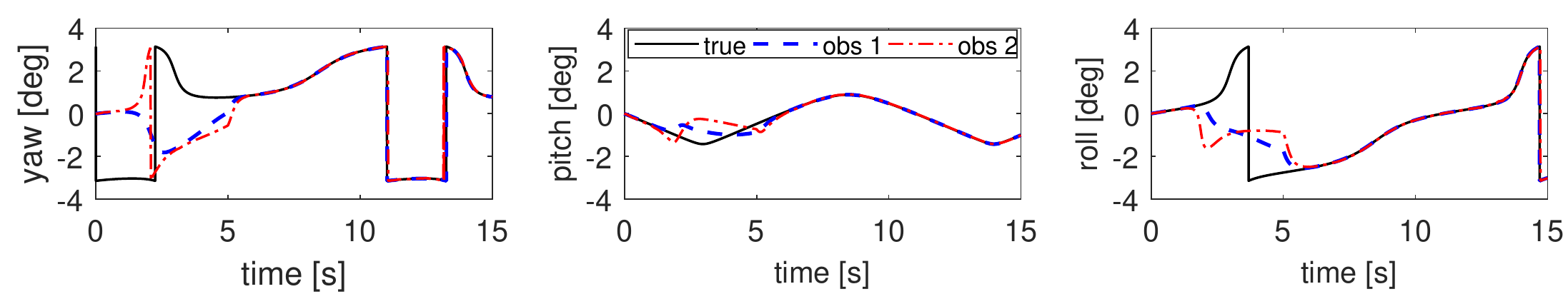}
\caption{Performance of the attitude observer in Proposition \ref{prop:obs1} (obs 1) and \ref{prop:obs2} (obs 2) with Euler angles (Example \ref{example:1})}
\label{fig:1}
\end{figure*}
\begin{figure}[htbp!]
\centering
\includegraphics[width = 0.95\linewidth]{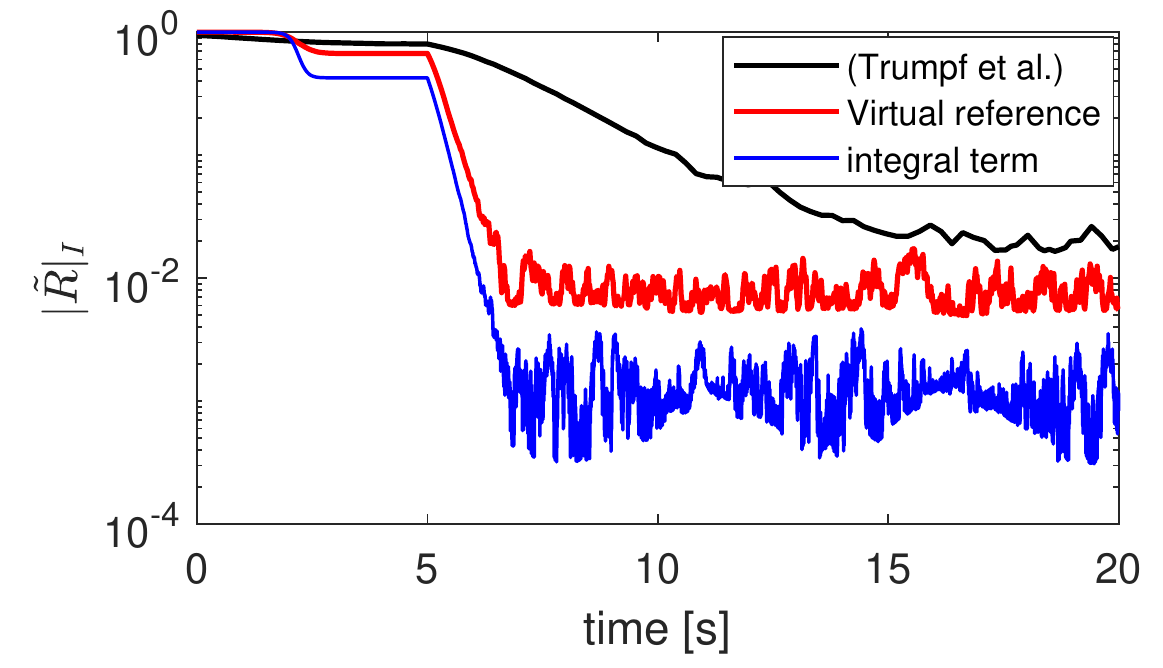}
\caption{Comparison of the norms of estimation errors $|\tilde R|_I$ among the proposed observers and \cite{TRUetal} (Example \ref{example:1})}
\label{fig:3}
\end{figure}

\end{example}
\begin{figure*}[htbp!]
\centering
\includegraphics[width = 0.99\linewidth]{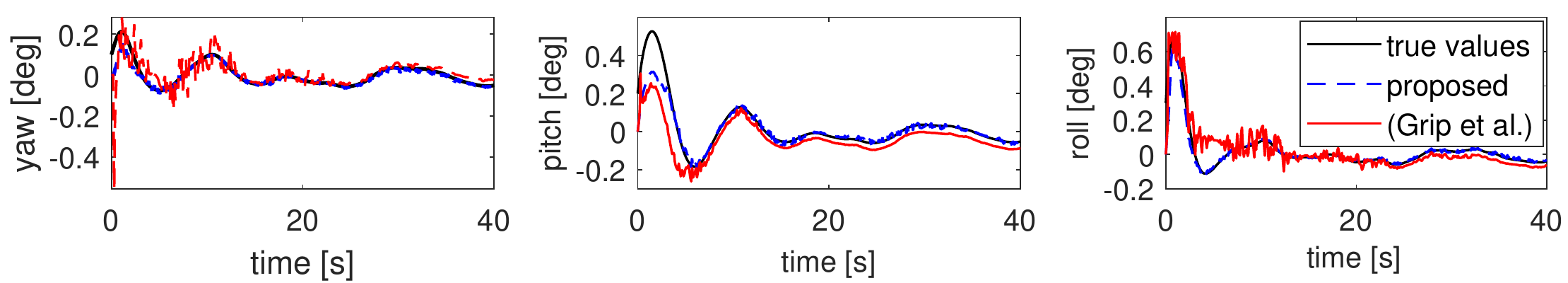}
\caption{Simulation results for helicopter attitude estimation using the observers in Proposition \ref{prop:obs1} and in \cite{GRIetalTAC} (Example \ref{example:2})}
\label{fig:7}
\end{figure*}
\begin{figure}[htbp!]
\centering
\includegraphics[width = 0.9\linewidth]{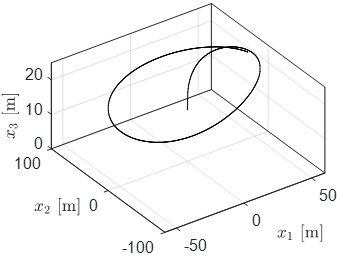}
\caption{The trajectory of the helicopter (Example \ref{example:2})}
\label{fig:8}
\end{figure}
\begin{figure}[htbp!]
\centering
\includegraphics[width = 0.95\linewidth]{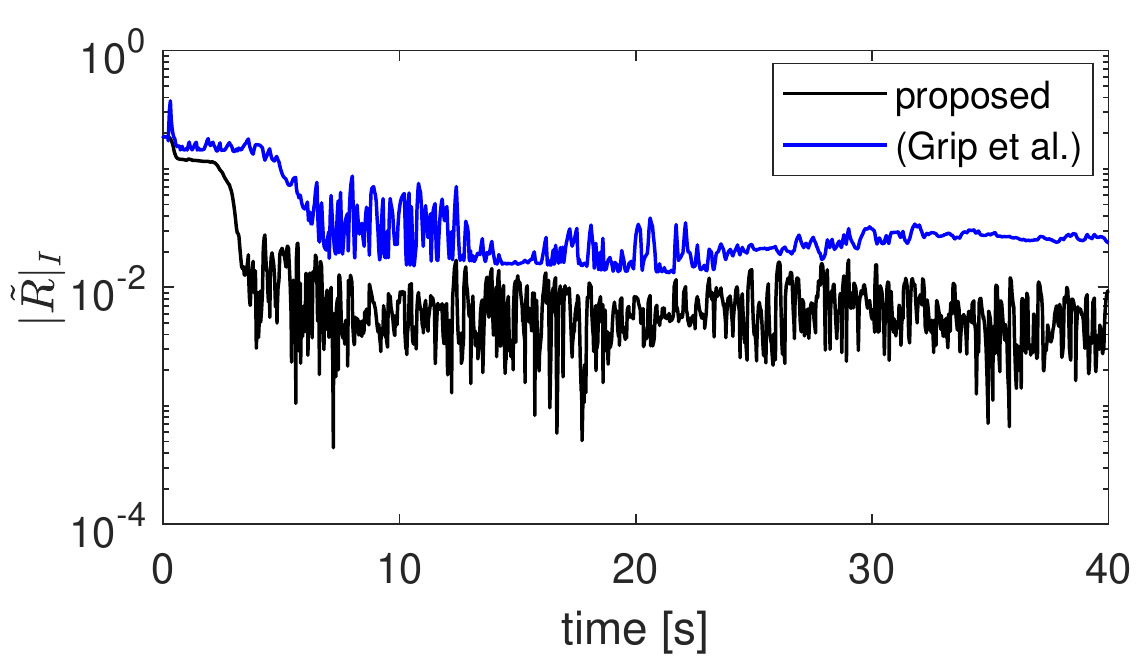}
\caption{Comparison of the norms of estimation errors $|\tilde R|_I$ among the proposed observer and the design in \cite{GRIetalTAC} (Example \ref{example:2})}
\label{fig:3}
\end{figure}

\begin{example}
\label{example:2}
In the second example, we consider the problem set in \cite{GRIetalTAC}, i.e., the vector being the highly time-varying acceleration of a helicopter. To be precise, a remotely controlled helicopter is equipped with accelerometers to detect the acceleration $y_{\tt B} = {}^{\tt B}a$ in the body-fixed frame, and the corresponding inertial acceleration ${}^{\tt I}a$ can be calculated from the GPS velocity $v$ using the relation $\dot v = {}^{\tt I} a$. Clearly, the relation ${}^{\tt B}a = R^\top {}^{\tt I}a$ holds true. Note that the ``dirty derivative'' is usually provided by the filtered approximation $H_1(p)[v]={\alpha p \over \alpha + p }[v]$ with the differential operator $p:={d\over dt}$ and selecting $\alpha>0$ some large parameter. It is widely known that such an operation yields the phase shift. In our problem set, in terms of Proposition \ref{prop:obs2}, we have the identity
$
Q(t) {}^{\tt B}a(t) = Q_c \dot v(t).
$
Applying the filter $H_2(p)[\cdot]={\alpha  \over \alpha + p }[\cdot]$ to the both sides -- thanks to $Q_c$ being constant -- we have
\begin{equation}
\label{filter}
H_2\big[Q {}^{\tt B}a \big] ~=~ Q_c H_1[v].
\end{equation}
By injecting the above into the proposed observer, we are able to overcome the phase-shifting issue caused by the filter $H_1(p)$, and there is no need to require the parameter $\alpha$ large to approximate the time derivative. Instead, a small $\alpha$ makes $H_2(p)$ behave as a low-pass filter, which is capable to attenuate noise significantly.

We consider the helicopter trajectory as shown in Fig. \ref{fig:8}. To make the simulation more realistic, the gyros and the accelerometer provide data at 100 Hz, and the GPS receiver is at 10 Hz -- all with high-frequency noise. Besides, in order to evaluate robustness of the proposed design, we consider acceleration bias from the senor, but do not make any compensation. The first observer design was implemented at 1000 Hz using the solver ``ODE 4 (Runge-Kutta)'' in Matlab/Simulink\textsuperscript{TM}, with $\gamma_{\tt I} =1,\gamma_{\tt P}=5,\alpha=1$, and the initial conditions $Q(0)= \hat Q_c(0) = I_3$. The simulation results are given in Fig. \ref{fig:7}, which illustrates its good robustness, though it brings additional errors in the steady-state stage. We compare it to the design in \cite{GRIetalTAC} using $H_1(p)$ with $\alpha=8$ approximate the differentiator. The phase lag from the filter leads to the offset in estimates observed in Fig. \ref{fig:7}. This effect can be reduced by increasing $\alpha$, but at the expense of higher sensitivity to noise. As discussed above, using \eqref{filter} the proposed design does not suffer from this issue.
\end{example}

%
%%%%%%%%%%%%%%%%%%%%%
\section{Concluding Remarks}
\label{sec7}
%%%%%%%%%%%%%%%%%%%%%
%

In this paper, we studied the observability and observer design for the attitude estimation problem with vectorial measurements. By translating the observation problem into one of on-line parameter identification, we provided the necessary and sufficient condition to the distinguishability for the dynamical model on $SO(3)$, which is complementary to the existing necessary conditions in the literature. As is shown later, though the resulting distinguishability condition is quite weak, we are still able to use it to derive a continuous-time attitude observer with almost global asymptotic stability guaranteed for the single vector case. Finally, simulation results demonstrated accurate estimation performance in the presence of measurement noise. 

%========================================

\appendix

%
%%%%%%%%%%%%
\section*{Proof of Corollary \ref{cor:modified}}
%%%%%%%%%%%%
%
\begin{proof}
Compared to Proposition \ref{prop:cond}, the only difference relies on the
second term, which corresponds to the third case in the proof of Proposition
\ref{prop:cond}. 

The modified condition assumes the existence of two indices $i \in \ell_1, j\in
\ell_2$ such that
$
 |g_i(t_1) \times [\Phi(0,t_2) b_j(t_2)]  |>0.
$
Since $\Phi(0,t_2) b_j(t_2) \in \mathbb{S}^2$, we have
\begequ
\label{gitimes}
 \Big|g_i(t_1) \times \big[R_0\Phi(0,t_2) b_j(t_2) \big]  \Big|>0
\endequ
if $R_0$ is not in an inadmissible initial set
\begequ
\label{cale}
\cale := \{R_0\in SO(3) ~|~ R_0 \mathbf{v} = \pm \mathbf{w} \}
\endequ
with $\mathbf{v} := \Phi(0,t_2) b_j(t_2)$ and $\mathbf{w} = g_i(t_1)$ both in $\mathbb{S}^2$. For a given rotational velocity $\omega$ and the references $b_j, g_i$, two of three Euler angles of the initial rotation matrix $R_0$ is uniquely determined by the equality in \eqref{cale}. Hence, the inadmissible initial set $\cale$ has zero Lebesgue measure in the group $SO(3)$. As a result, we guarantee the condition \eqref{case:3} from the modified assumption almost surely. 
\end{proof}

%
%%%%%%%%%%%%
\section*{Proof of Corollary \ref{cor:subset}}
%%%%%%%%%%%%
%

\begin{proof}
Since those two terms in \eqref{cond:trumpf} are non-negative, is should
satisfy at least one of the following cases:
\begin{align}
\mbox{(i)}& &\lambda_2 \left( \sum_{i\in \ell_1}\int_0^T g_i(s) g_i^\top(s) ds
\right) >0 \label{case:1}
\\
\mbox{(ii)}& &
\left\| 
\int_0^T \sum_{j\in \ell_2} \left(\omega_\times b_j(s) + {d\over ds}b_j(s) 
\right)ds
\right\|>0. \label{case:2}
\end{align}
In the case (i), note that for a single vector $g_i(t)\in \mathbb{S}^2$, the matrix $g_i(t)g_i^\top(t)$ has rank one at any instance $t\ge 0$. Hence, a necessary
condition to \eqref{case:1} is the existence of $t_1,t_2\ge 0$ and $i,l \in
\ell_1$ ($i$ and $l$ may be the same) such that
\begin{equation}
 \lambda_2 \Big(g_i(t_1)g_i^\top(t_1) + g_l(t_2) g_l^\top(t_2) \Big) >0,
\end{equation}
which implies $\sum_{\substack{i,l \in \ell_1}} 
  | g_i(t_1)\times g_l(t_2)| >0$, thus guaranteeing the condition
\eqref{cond:1}.
  
For the case (ii), we consider (piecewisely) smooth outputs $y_{{\tt I},i}$
with $i\in \ell_2$. Its dynamics is given by
\begin{equation}
\label{dot:yi}
\dot{y}_{{\tt I},i} ~ = ~ R (\omega_\times b_i + \dot b_i).
\end{equation}
A necessary condition to \eqref{case:2} is that there exist $j$ and a moment
$t_1>0$ such that 
$$
\omega(t_1)_\times b_j(t_1) + \dot b_j(t_1) \neq 0
\quad \implies \quad
\dot y_{{\tt I},j}(t_1) \neq 0.
$$
Let us select a sufficiently small $\Delta t>0$, and define $t_2:= t_1 + \Delta
t$. It yields
$$
y_{{\tt I},j}(t_2) ~=~ y_{{\tt I},j}(t_1) + \dot y_{{\tt I},j}(t_1) \Delta t +
o(\Delta t^2)
$$
in which $o(\Delta t^2)$ represents the high-order remainder term, with the
\emph{constraint} $y_{{\tt I},j}(t_2) \in \mathbb{S}^2$. Now, we show $y_{{\tt
I},j}(t_1)\times {d\over dt} y_{{\tt I},j}(t_1) \neq 0$ by contradiction. If this cross product
is equal to zero, invoking ${d\over dt}y_{{\tt I},j}(t_1)\neq 0$, we have
$
\dot y_{{\tt I},j}(t_1) = a y_{{\tt I},j}(t_1)
$
for some non-zero $a \in \rea$. Then, we have 
$$
\begin{aligned}
|y_{{\tt I},j}(t_2)| & ~=~ \Big|(1+a\Delta t) y_{{\tt I},j}(t_2) + o (\Delta t^2)
\Big|
\\
& ~=~  |1+a\Delta t| + o (\Delta t^2),
\end{aligned}
$$
which contradicts with the fact $y_{{\tt I},j}(t_2) \in \mathbb{S}^2$. As a consequence,
we obtain that $y_{{\tt I},j}(t_1)\times y_{{\tt I},j}(t_2)\neq 0$. Invoking
the equivalence \eqref{case:2-1}, we have
$$
\eqref{case:2}\quad \implies \quad
\sum_{j,k \in \ell_2}
  \Big| b_j(t_1)_\times \Phi(t_1,t_2) b_k(t_2) \Big| >0,
$$
thus verifying \eqref{cond:1}. It completes the proof.
\end{proof}

\bibliographystyle{abbrv}
\bibliography{reference}

\begin{thebibliography}{10}

\bibitem{BAHNAM}
S.~Bahrami and M.~Namvar.
\newblock Global attitude estimation using single delayed vector measurement
  and biased gyro.
\newblock {\em Automatica}, 75:88--95, 2017.

\bibitem{BATetal}
P.~Batista, C.~Silvestre, and P.~Oliveira.
\newblock A {GES} attitude observer with single vector observations.
\newblock {\em Automatica}, 48(2):388--395, 2012.

\bibitem{BERTAYauto19}
S.~Berkane and A.~Tayebi.
\newblock Attitude estimation with intermittent measurements.
\newblock {\em Automatica}, 105:415--421, 2019.

\bibitem{BER}
P.~Bernard.
\newblock {\em Observer {Design} for {Nonlinear} {Systems}}.
\newblock Lecture {Notes} in {Control} and {Information} {Sciences}. Springer
  International Publishing, 2019.

\bibitem{BES}
G.~Besan{\c{c}}on.
\newblock {\em Nonlinear Observers and Applications}, volume 363.
\newblock Springer, 2007.

\bibitem{BOBetal}
A.~Bobtsov, B.~Yi, R.~Ortega, and A.~Astolfi.
\newblock Generation of new exciting regressors for consistent on-line
  estimation of unknown constant parameters.
\newblock {\em IEEE Trans. Autom. Control}, 2022.

\bibitem{CRAetal}
J.~L. Crassidis, F.~L. Markley, and Y.~Cheng.
\newblock Survey of nonlinear attitude estimation methods.
\newblock {\em J. Guid. Control Dyn.}, 30(1):12--28, 2007.

\bibitem{GRIetalTAC}
H.~F. Grip, T.~I. Fossen, T.~A. Johansen, and A.~Saberi.
\newblock Attitude estimation using biased gyro and vector measurements with
  time-varying reference vectors.
\newblock {\em IEEE Trans. Autom. Control}, 57(5):1332--1338, 2011.

\bibitem{KINWHI}
J.~C. Kinsey and L.~L. Whitcomb.
\newblock Adaptive identification on the group of rigid-body rotations and its
  application to underwater vehicle navigation.
\newblock {\em IEEE Trans. Robot.}, 23(1):124--136, 2007.

\bibitem{MAHetal08}
R.~Mahony, T.~Hamel, and J.-M. Pflimlin.
\newblock Nonlinear complementary filters on the special orthogonal group.
\newblock {\em IEEE Trans. Autom. Control}, 53(5):1203--1218, 2008.

\bibitem{MARSAR}
P.~Martin and I.~Sarras.
\newblock Partial attitude estimation from a single measurement vector.
\newblock In {\em Proc. IEEE Conf. Control Tech. Appl.}, pages 1325--1331,
  2018.

\bibitem{NAMSAF}
M.~Namvar and F.~Safaei.
\newblock Adaptive compensation of gyro bias in rigid-body attitude estimation
  using a single vector measurement.
\newblock {\em IEEE Trans. Autom. Control}, 58(7):1816--1822, 2013.

\bibitem{ORTetalAUT}
R.~Ortega, A.~Bobtsov, N.~Nikolaev, J.~Schiffer, and D.~Dochain.
\newblock Generalized parameter estimation-based observers: Application to
  power systems and chemical--biological reactors.
\newblock {\em Automatica}, 129, 2021.
\newblock Art. no. 109635.

\bibitem{ORTetalscl}
R.~Ortega, A.~Bobtsov, A.~Pyrkin, and S.~Aranovskiy.
\newblock A parameter estimation approach to state observation of nonlinear
  systems.
\newblock {\em Syst. Control Lett.}, 85:84--94, 2015.

\bibitem{ORTetalReview}
R.~Ortega, V.~Nikiforov, and D.~Gerasimov.
\newblock On modified parameter estimators for identification and adaptive
  control: {A} unified framework and some new schemes.
\newblock {\em Annu. Rev. Control}, 50:278--293, 2020.

\bibitem{REIetal}
J.~Reis, P.~Batista, P.~Oliveira, and C.~Silvestre.
\newblock Attitude, body-fixed earth rotation rate, and sensor bias estimation
  using single observations of direction of gravitational field.
\newblock {\em Automatica}, 125, 2021.
\newblock Art. no. 109475.

\bibitem{SAL}
S.~Salcudean.
\newblock A globally convergent angular velocity observer for rigid body
  motion.
\newblock {\em IEEE Trans. Autom. Control}, 36(12):1493--1497, 1991.

\bibitem{SASBOD}
S.~Sastry and M.~Bodson.
\newblock {\em Adaptive Control: Stability, Convergence And Robustness}.
\newblock Courier Corporation, 2011.

\bibitem{SHUOH}
M.~D. Shuster and S.~D. Oh.
\newblock Three-axis attitude determination from vector observations.
\newblock {\em J. Guid. Control Dyn.}, 4(1):70--77, 1981.

\bibitem{TRUetal}
J.~Trumpf, R.~Mahony, T.~Hamel, and C.~Lageman.
\newblock Analysis of non-linear attitude observers for time-varying reference
  measurements.
\newblock {\em IEEE Trans. Autom. Control}, 57(11):2789--2800, 2012.

\bibitem{WAH}
G.~Wahba.
\newblock A least squares estimate of satellite attitude.
\newblock {\em SIAM Review}, 7(3):409--409, 1965.

\bibitem{WANetal}
L.~Wang, R.~Ortega, A.~Bobtsov, J.~G. Romero, and B.~Yi.
\newblock Identifiability implies robust, globally exponentially convergent
  on-line parameter estimation: Application to model reference adaptive
  control.
\newblock {\em arXiv preprint arXiv:2108.08436}, 2021.

\bibitem{YIJIN}
B.~Yi, C.~Jin, and I.~R. Manchester.
\newblock Globally convergent visual-feature range estimation with biased
  inertial measurements.
\newblock {\em arXiv preprint arXiv:2112.12325}, 2021.

\bibitem{YIetalCDC}
B.~Yi, C.~Jin, L.~Wang, G.~Shi, and I.~R. Manchester.
\newblock An almost globally convergent observer for visual {SLAM} without
  persistent excitation.
\newblock In {\em 60th Proc. IEEE Conf. Decis. Control}, pages 5441--5446,
  2021.

\bibitem{YIetalTAC19}
B.~Yi, R.~Ortega, and W.~Zhang.
\newblock On state observers for nonlinear systems: A new design and a unifying
  framework.
\newblock {\em IEEE Trans. Autom. Control}, 64(3):1193--1200, 2018.

\bibitem{ZLOFOR}
D.~E. Zlotnik and J.~R. Forbes.
\newblock Nonlinear estimator design on the special orthogonal group using
  vector measurements directly.
\newblock {\em IEEE Trans. Autom. Control}, 62(1):149--160, 2016.

\end{thebibliography}

\end{document}